\documentclass[a4paper]{article}
\usepackage[left=3cm,right=3cm,top=4cm,bottom=3cm]{geometry}
\usepackage{microtype}

\usepackage{graphicx}
\usepackage{latexsym}
\usepackage{amsthm,amsmath,amssymb}
\usepackage{enumerate}
\usepackage{algorithm}
\usepackage{algorithmic}
\usepackage{subfig}
\usepackage{comment}
\usepackage{xspace}
\usepackage{url}
\usepackage{authblk}
\usepackage{flushend}

\newcommand*{\ie}{i.e.\@\xspace}
\newcommand*{\cf}{cf.\@\xspace}

\makeatletter
\newcommand*{\etc}{%
    \@ifnextchar{.}%
        {etc}%
        {etc.\@\xspace}%
}
\makeatother

\newcommand{\figref}[1]{Figure~\ref{#1}}
\newcommand{\Figref}[1]{Figure~\ref{#1}}

\def\numberwang{\textsc{Zig-Zag Numberlink}\xspace}
\def\Underline{\setbox0\hbox\bgroup\let\\\endUnderline}
\def\endUnderline{\vphantom{y}\egroup\smash{\underline{\box0}}\\}
\def\|{\verb|}
\newtheorem{theorem}{Theorem}

\newtheorem{definition}{Definition}

\title{Zig-Zag Numberlink is NP-Complete}

\author[1]{Aaron Adcock}
\affil[1]{\small Stanford University\\
Palo Alto, CA, USA\\
\texttt{aadcock@stanford.edu}}

\author[2]{Erik D.~Demaine}
\author[2]{Martin L.~Demaine}
\affil[2]{\small Massachusetts Institute of Technology\\
Cambridge, MA, USA\\
\texttt{\{edemaine,mdemaine\}@mit.edu}}

\author[3]{Michael P.~O'Brien}
\affil[3]{North Carolina State University\\
Raleigh, NC, USA\\
\texttt{\{mpobrie3,blair\_sullivan\}@ncsu.edu}}

\author[4]{Felix Reidl}
\author[4]{Fernando S\'a{}nchez Villaamil}
\affil[4]{\small RWTH Aachen University\\ 
Aachen, Germany\\
\texttt{\{reidl,fernando.sanchez\}@cs.rwth-aachen.de}}

\author[3]{Blair D.~Sullivan}

\begin{document}

\maketitle

\begin{abstract}
  When can $t$ terminal pairs in an $m \times n$ grid be connected by
  $t$ vertex-disjoint paths that cover all vertices of the grid?
  We prove that this problem is NP-complete.
  Our hardness result can be compared to two previous NP-hardness proofs:
  Lynch's 1975 proof without the ``cover all vertices'' constraint,
  and Kotsuma and Takenaga's 2010 proof when the paths are restricted to
  have the fewest possible corners within their homotopy class.
  The latter restriction is a common form of the famous Nikoli puzzle
  \emph{Numberlink}; our problem is another common form of Numberlink,
  sometimes called \emph{Zig-Zag Numberlink} and popularized by the
  smartphone app \emph{Flow Free}.
\end{abstract}

\section{Introduction}
Nikoli is a famous Japanese publisher of pencil-and-paper puzzles,
best known world-wide for its role in popularizing Sudoku puzzles
\cite{SudokuHistory}.  Nikoli in fact publishes whole ranges of such puzzles,
following rules of their own and others' inventions; see \cite{NikoliWeb}.
The meta-puzzle for us theoretical computer scientists is to study the
computational complexity of puzzles, characterizing them as polynomially
solvable, NP-complete, or harder \cite{GamesBook}.
In the case of Nikoli puzzles, nearly every family has been proved NP-complete:
Country Road \cite{YajilinandCountryRoad}, Corral \cite{CorralNP},
Fillomino \cite{ABunch}, Hiroimono \cite{Hiroimono}, Hashiwokakero
\cite{Hashiwokakero}, Heyawake \cite{Heyawake}, Kakuro \cite{Kakuro},
Kurodoko \cite{Kurodoko}, Light Up \cite{LightUp}, Masyu \cite{Masyu},
Nurikabe (\cite{Nurikabe1,Nurikabe2,Nurikabe3}), Shakashaka
\cite{Shakashaka}, Slither Link (\cite{SlitherLink,ABunch}), Sudoku
(\cite{Sudoku1,Sudoku2,ABunch}), Yajilin \cite{YajilinandCountryRoad},
and Yosenabe \cite{Yosenabe}.

In this paper, we study the computational complexity of one Nikoli family of
puzzles called \emph{Numberlink} (also known as Number Link, Nanbarinku,
Arukone, and Flow).
A Numberlink puzzle consists of an $m \times n$ grid of unit squares,
some of which contain numbers, which appear in pairs.
The goal of the player is to connect corresponding pairs of numbers
by paths that turn only at the center of grid squares, do not cross any other
numbered squares, and do not cross any other paths.
Furthermore, in most versions of the puzzle, the paths should together visit
every grid square; in many puzzles, this constraint is in fact forced
by any otherwise valid solution.%
\footnote{This coverage constraint is absent from Nikoli's Numberlink
  website (\url{http://www.nikoli.co.jp/en/puzzles/numberlink.html}),
  but this may simply be an omission.}
Figure~\ref{fig:example_instance} shows a simple example.

\begin{figure}[tb]
  \centering
  \subfloat[Puzzle]{\includegraphics[width=0.35\columnwidth]{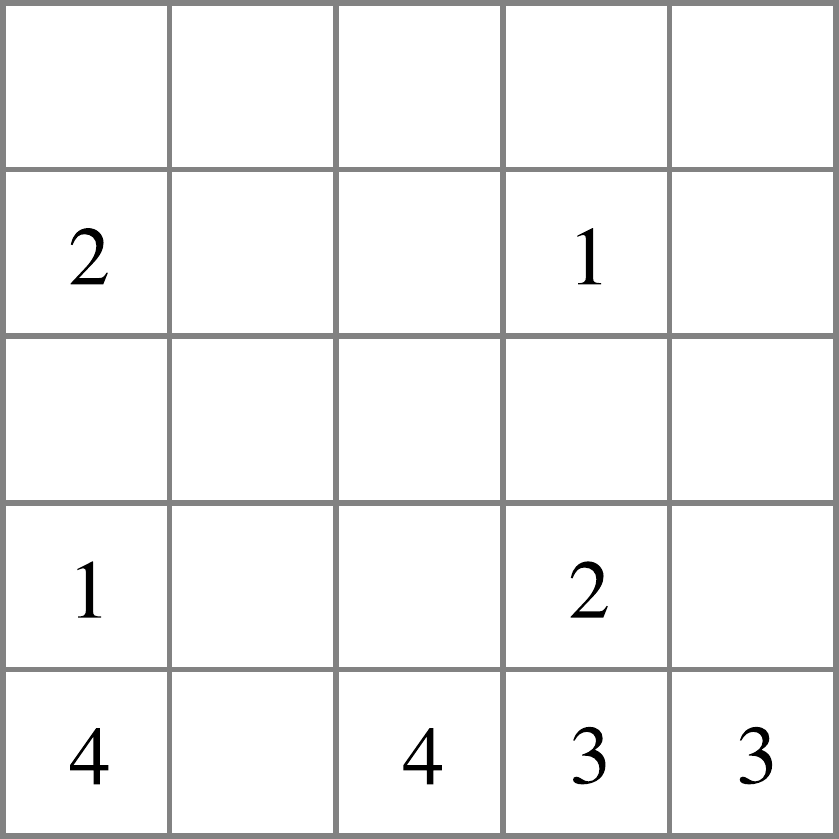}}
  \hspace{30pt}
  \subfloat[Solution]{\includegraphics[width=0.35\columnwidth]{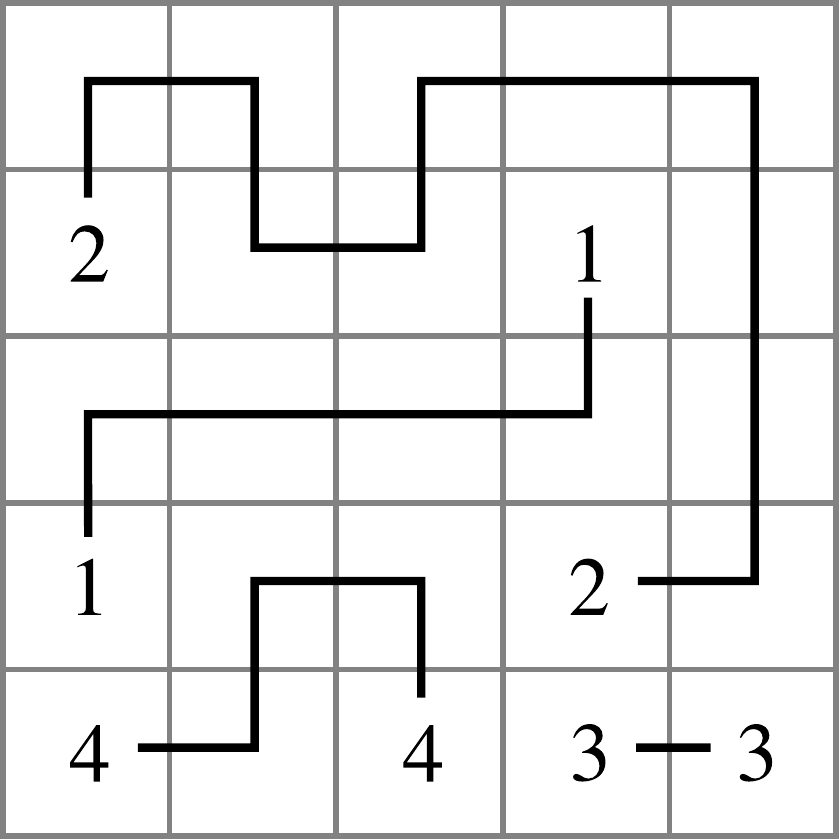}}
  \caption{\label{fig:example_instance}
    Sample $5\times 5$ instance with $4$ terminal pairs.}
\end{figure}

Although popularized and named by Nikoli, the history of Numberlink-style
puzzles is much older \cite{NumberlinkPegg,NumberlinkHistory}.
The earliest known reference is an 1897 column by Sam Loyd \cite{loyd1897}
(most famous for popularizing the 15 Puzzle in 1891 \cite{15puzzle}).
Figure~\ref{SamLoyd} reproduces his puzzle, called ``The Puzzled Neighbors''.
While the puzzle statement does not require visiting every square,
his solution visits most of the squares, and a small modification to his
solution visits all of the squares.
The same puzzle later appeared in Loyd's famous puzzle book \cite{loydbook}.
Another early Numberlink-style puzzle is by Dudeney in his famous 1931 
puzzle book \cite{dudeney1931536}.  His puzzle is $8 \times 8$ with five
terminal pairs, and while the puzzle statement does not require visiting every
square, his solution does.

Some Numberlink puzzles place an additional restriction on paths,
which we call \emph{Classic Numberlink} (following the Numberlink Android app).
Informally stated, restricted paths should not have ``unnecessary'' bends.
Although we have not seen a formal definition, based on several examples,
we interpret this restriction to mean that each path uses the fewest possible
turns among all paths within its homotopy class (i.e., according to which
other obstacles and paths it loops around).
This version of Numberlink has already been shown NP-complete
\cite{NumberlinkNP}.

In this paper, we analyze the unrestricted version, which we call
\emph{Zig-Zag Numberlink} (again following the Numberlink Android app).
Informally, this style of puzzle allows links to zig-zag arbitrarily to fill
the grid.  Personally, we find this formulation of the problem more natural,
given its connection to both vertex-disjoint paths and Hamiltonicity in graphs.
When we only allow for one pair of terminals this problem reduces to finding
a Hamiltonian path in a grid given fixed start and end points, which is known 
to be solvable in polynomial time~\cite{itai1982hamilton}.
Unfortunately, for the case with several terminal pairs, 
as shown in Section~\ref{NP-hardness}, the hardness proof of
\cite{NumberlinkNP} does not (immediately) apply to Zig-Zag Numberlink.
Nonetheless, we construct a very different and intricate NP-hardness proof,
inspired by an early NP-hardness proof from 1975 for vertex-disjoint paths
\cite{Lynch-1975}.

\section{Definitions}

Next we formally define the puzzle \numberwang.

\begin{definition}
  A \emph{board} $B_{m,n}$ is a rectangular grid of $mn$ equal sized
  \emph{squares} arranged into $m$ rows and $n$ columns.  We will
  identify the squares with an ordered pair consisting of their column and row
  position. The top left corner is defined to be position $(1,1)$.
\end{definition}


\begin{definition}
  A \emph{terminal pair} is a pair of distinct squares. An
  \emph{instance of} \numberwang is a tuple $\mathcal{F} =
  (B_{m,n},\mathcal{T})$ where $B_{m,n}$ is a $m \times n$ board and
  $\mathcal{T} = \{(T_{1},T'_{1}),\ldots, (T_t, T'_t)\}$ a set of
  terminal pairs. All terminals are distinct, that is,
  any square of the board may contain at most one terminal.
\end{definition}


\begin{definition} 
  Two squares $(x,y)$ and $(p,q)$ are \emph{adjacent} if either $x = p$ and $|y-q| = 1$, or $y = q$ and $|x-p| = 1.$ 
  A sequence of squares $P=s_{1},\ldots,s_{k}$ is a \emph{path} of length $k$ if $s_i$ is adjacent to $s_{i+1}$ for $i \in [1,k-1]$ and $s_i \neq s_j$ for all $i \neq j$. Two squares are \emph{linked} in $P$ if they appear successively in $P$.
\end{definition}

\begin{figure}[H]
  \centering
  \subfloat[Puzzle]{\includegraphics[angle=-1,width=0.45\columnwidth]{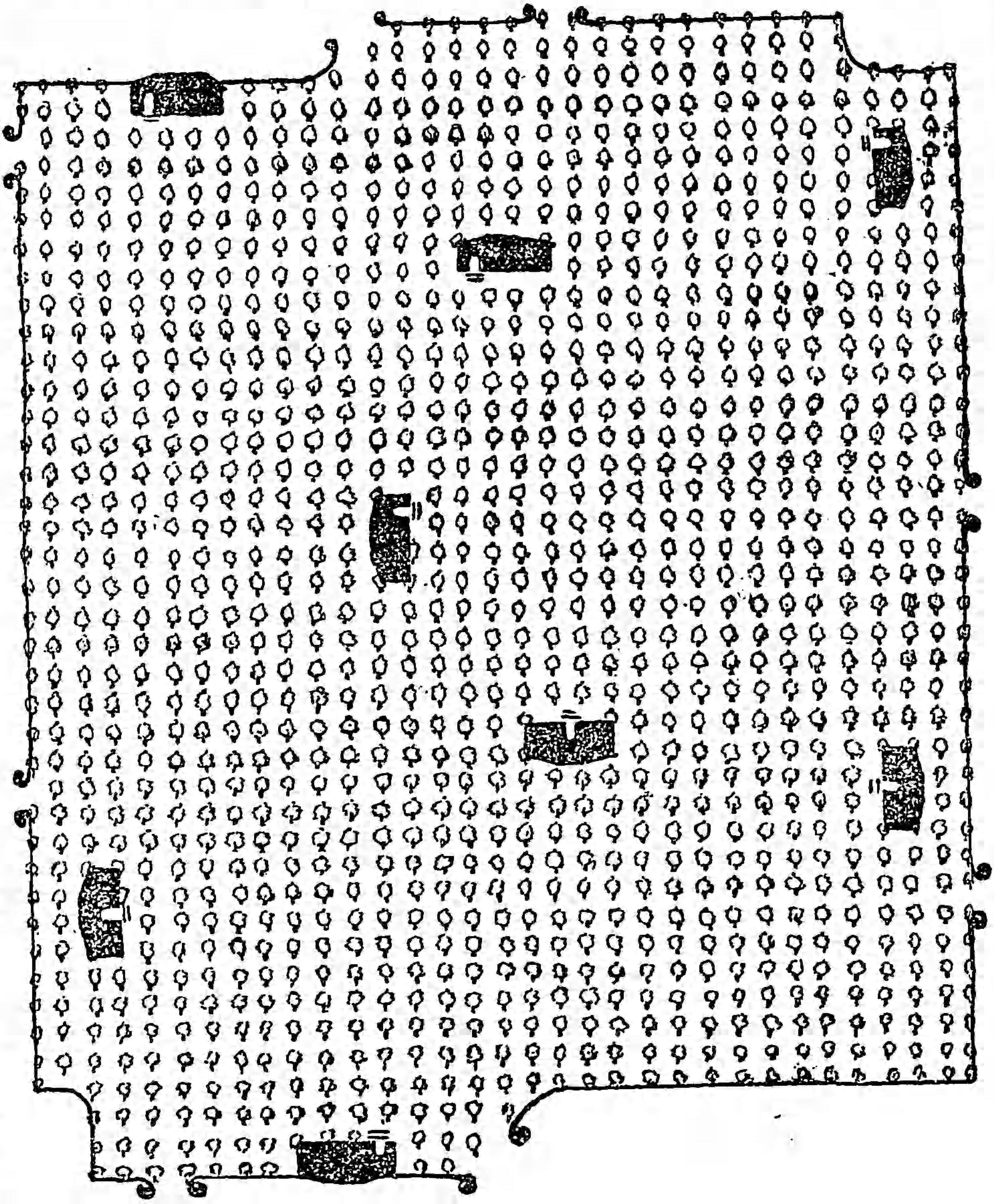}}%
  \hspace{30pt}
  \subfloat[Solution]{\includegraphics[width=0.45\columnwidth]{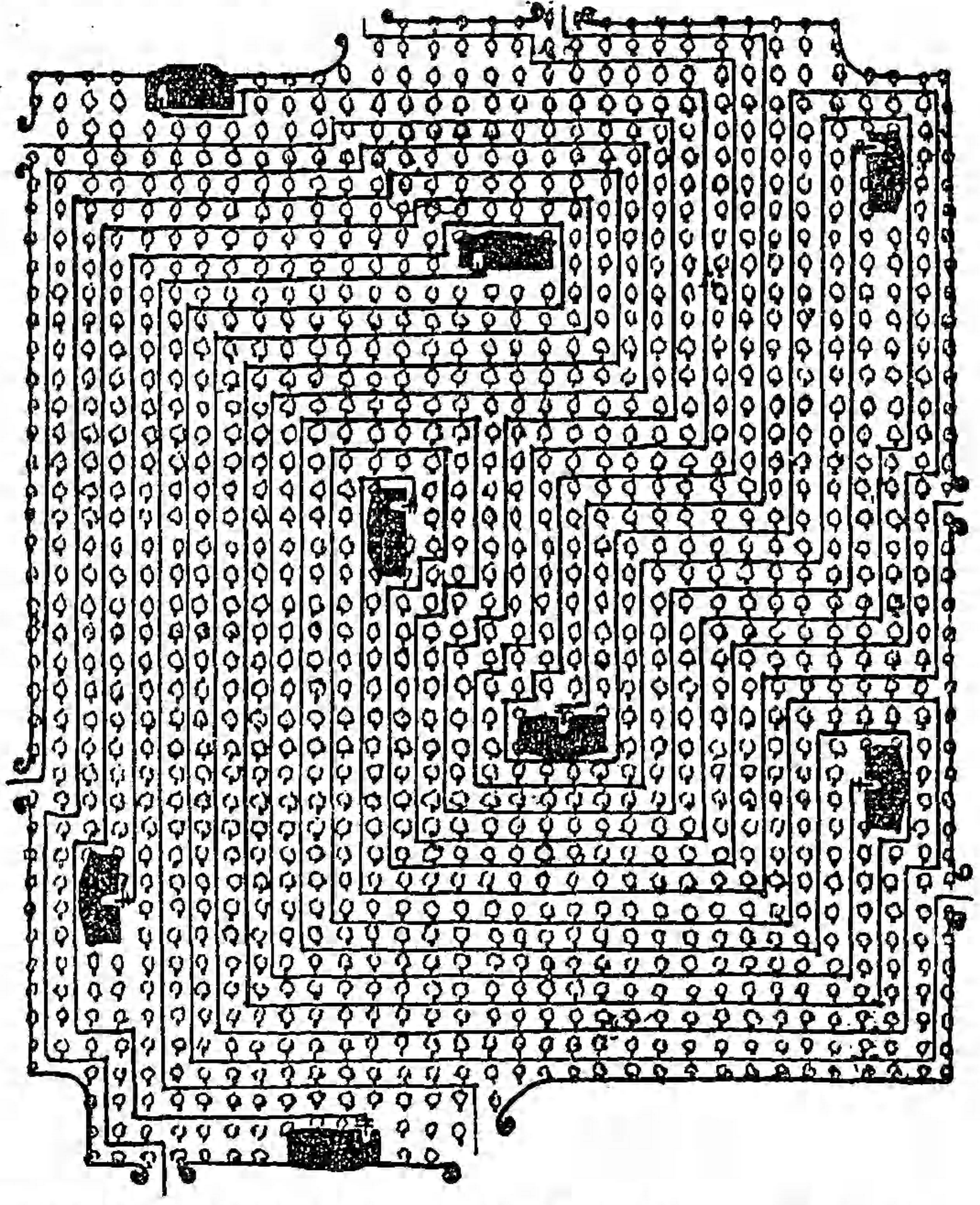}}%
  \caption{\label{SamLoyd}
    Sam Loyd's ``The Puzzled Neighbors'' from 1897 \cite{loyd1897}.
    Each house should be connected by a path to the gate it directly faces,
    by noncrossing paths.  Scans from
    \protect\url{http://bklyn.newspapers.com/image/50475607/} and
    \protect\url{http://bklyn.newspapers.com/image/50475838/}.}
\end{figure}

\begin{definition}
  A \emph{solution} to a \numberwang instance $\mathcal F = (B,\mathcal T)$ is a set
  of paths $\mathcal S = \{P_1,\dots, P_t\}$, where
  $P_i=s_{i,1},\dots,s_{i,k_i}$, so that:
  \begin{enumerate}[(i)]
  \item Every terminal pair $(T_i, T'_i)$ is connected by path $P_i$,
    \ie~$s_{i,1} = T_i$ and $s_{i,k_i} = T'_i$.
  \item Each square in $B$ is contained in exactly one path in $\mathcal S$.
  \end{enumerate}

\end{definition}

We call an instance of \numberwang \emph{solvable} if there exists a solution
and \emph{unsolvable} otherwise. We say that two squares are \emph{linked} by
a solution $S$ if they are linked in some path $P \in \mathcal S$. In the case of $t=1$,
we will abuse the above notation and identify the solution by a single path.

%

Finally, we will talk about \emph{parity}. To illustrate this concept,
we will color the squares alternating black and white as on a checkerboard (\cf~\figref{fig:checkerboard}) and assume that $(1,1)$ is colored black. The important
aspect of parity is that any path of a solution necessarily alternates between
white and black squares. In the construction of our gadgets, parity helps
to ensure that a combination of gadgets still allows a solution. As it will
turn out, parity along with terminal position is crucial to determine
whether instances with only one terminal pair are solvable.

\begin{figure}[tb]
\centering
\begin{minipage}{.48\textwidth}
  \centering
  \vspace*{-30pt}\includegraphics[width=.8\linewidth]{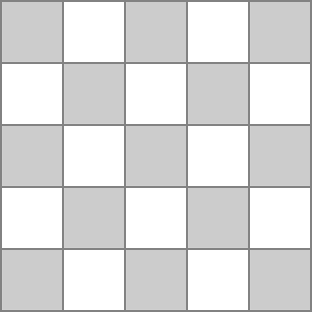}
  \captionof{figure}{Squares colored to illustrate parity.}
  \label{fig:checkerboard}
\end{minipage}%
\hspace{10pt}
\begin{minipage}{.48\textwidth}
  \centering
  \includegraphics[width=.8\linewidth]{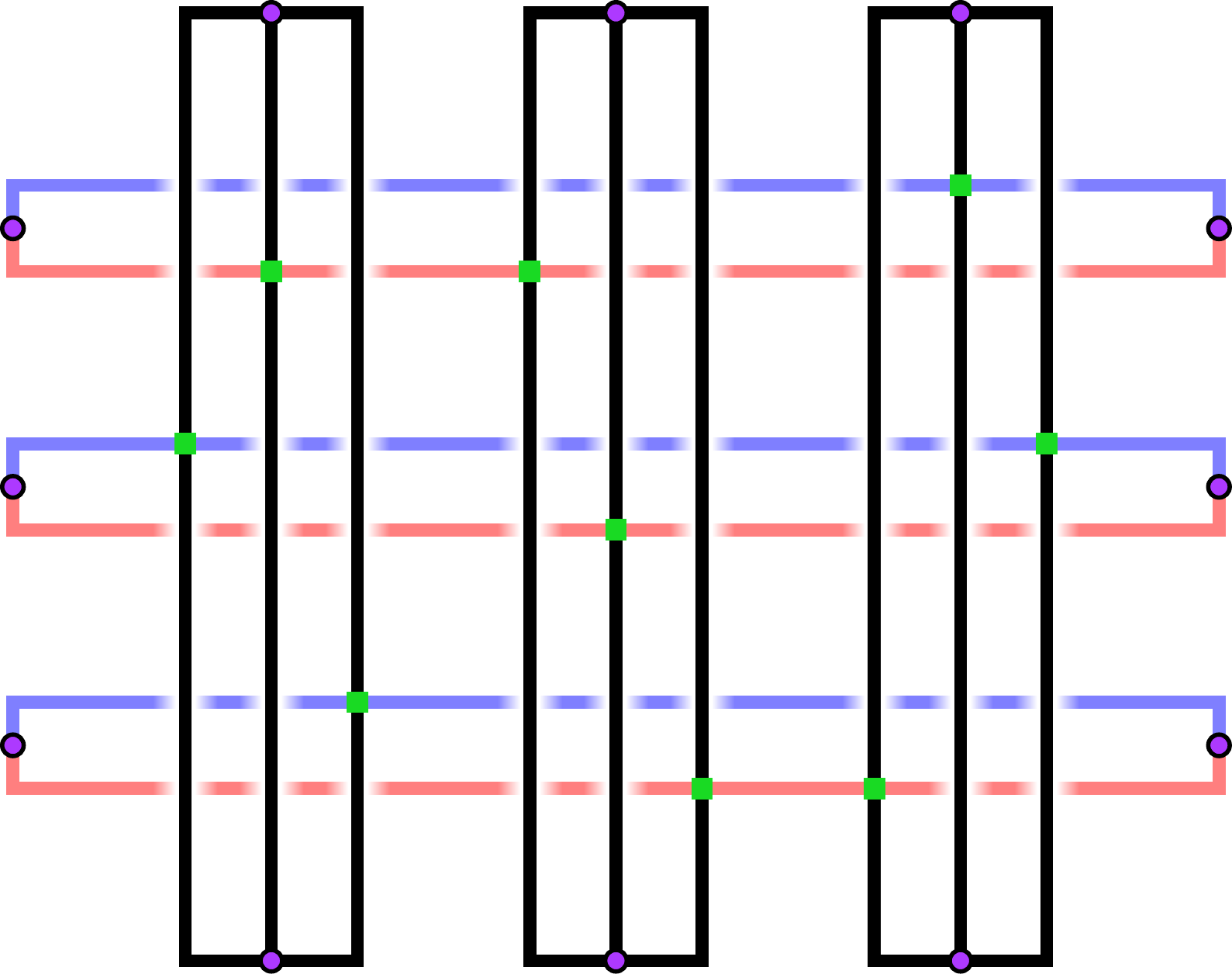}
  \captionof{figure}{High-level sketch of NP-hardness proof: two paths
    per variable, three paths per clause, and crossover gadgets (faded
    underpasses).  Circles indicate terminals; squares indicate actual
    crossings.  Whitespace indicates obstacles where paths cannot go.}
  \label{hardness overview}
\end{minipage}
\end{figure}


\section{NP-hardness}
\label{NP-hardness}
Our NP-hardness reduction for \numberwang mimics the structure of a very early NP-hardness
proof for vertex-disjoint paths among $k$ terminal pairs in grid graphs
by Lynch \cite{Lynch-1975}.  The important differences are that Lynch's
reduction (1)~allows obstacles (untraversable squares), and
(2)~does not require every traversable square to be covered by some path.
The first issue is relatively easy to deal with because terminals serve as
obstacles for all other paths.
For the second issue, we replace all of Lynch's gadgets with more complicated
gadgets to make it possible to cover every square of the grid in all cases.

\begin{theorem}
\numberwang is NP-complete for $k$ terminal pairs in an $n \times n$ square.
\end{theorem}

\begin{proof}
  \numberwang is in NP because the solution paths can be expressed in $O(n^2 \log k)$ space
  and checked in the same amount of time.

  To prove NP-hardness, we reduce from 3SAT.
  \Figref{hardness overview} illustrates the high-level picture.
  We construct one terminal pair $(v'_i,v''_i)$ for each variable $v_i$,
  and two candidate paths $V_{i,-},V_{i,+}$ for connecting this pair,
  $V_{i,-}$ representing the false setting and
  $V_{i,+}$ representing the true setting.
  We construct one terminal pair $(c'_j,c''_j)$ for each clause $c_j$,
  and three candidate paths $C_{j,1},C_{j,2},C_{j,3}$ for connecting this pair,
  one per literal in the clause $c_j = c_{j,1} \vee c_{j,2} \vee c_{j,3}$.
  Effectively, for each clause literal $c_{j,k} = \pm v_i$ say,
  the corresponding literal path $C_{j,k}$ intersects just the
  variable path $V_{i,\mp}$ corresponding to the setting
  that does \emph{not} satisfy the clause literal.
  Thus, setting the variable in this way blocks that clause path, and
  a clause must have one of its paths not blocked in this way
  (corresponding to satisfaction).
  It follows that any noncrossing choice of paths corresponds to
  a satisfying assignment of the 3SAT instance.
  
  The reality is more complicated because, in a square grid, all variable paths
  will intersect all clause paths.  However, we can simulate the
  nonintersection of two paths using the \emph{crossover gadget}
  shown in \figref{hardness crossover}.
  The idea is to split the two variable paths $V_{i,-},V_{i,+}$
  into two classes of variable paths, top and bottom,
  with nonintersection of the paths forcing alternation
  between top and bottom classes.
  A variable starting with a top path corresponds to a false setting
  (darker in the figure),
  and starting with a bottom path corresponds to a true setting
  (lighter in the figure).
  The crossover gadget must work when the vertical clause path is either
  present (chosen to satisfy the clause)
  or absent (having chosen a different of the three paths),
  resulting in four total cases.

  We simulate obstacles between paths using many pairs of terminals at unit
  distance (black in the figure).  These \emph{obstacle pairs} can be connected
  by a unit-length edge (as drawn in black) or by a longer path (drawn
  lighter, as a replacement for the black path, though the black path is
  still drawn).  Such longer paths can only prevent choices for the
  variable/clause paths, so any nonintersecting set of paths still solves
  the 3SAT instance.  Furthermore, by connecting the obstacle pairs by the
  lighter longer paths illustrated in the figure, in each of the four cases,
  we can turn any solution to the 3SAT instance into a valid solution to
  \numberwang (which in particular visits every square).\looseness-1%
  \footnote{Our careful and deliberate placement and orientation of obstacle
    pairs to enable such ``filling'' in all cases is the main novelty to
    our proof.  Lynch's crossover gadget has the same nonobstructed paths
    as our \figref{hardness crossover blank}, but as his proof allowed
    obstacles, lacked the complexity of obstacle pairs and the four cases.}

\begin{figure}[H]
  \centering
  \subfloat[Crossing (preventing) true, while still allowing false.]
  {\includegraphics[width=.48\columnwidth]{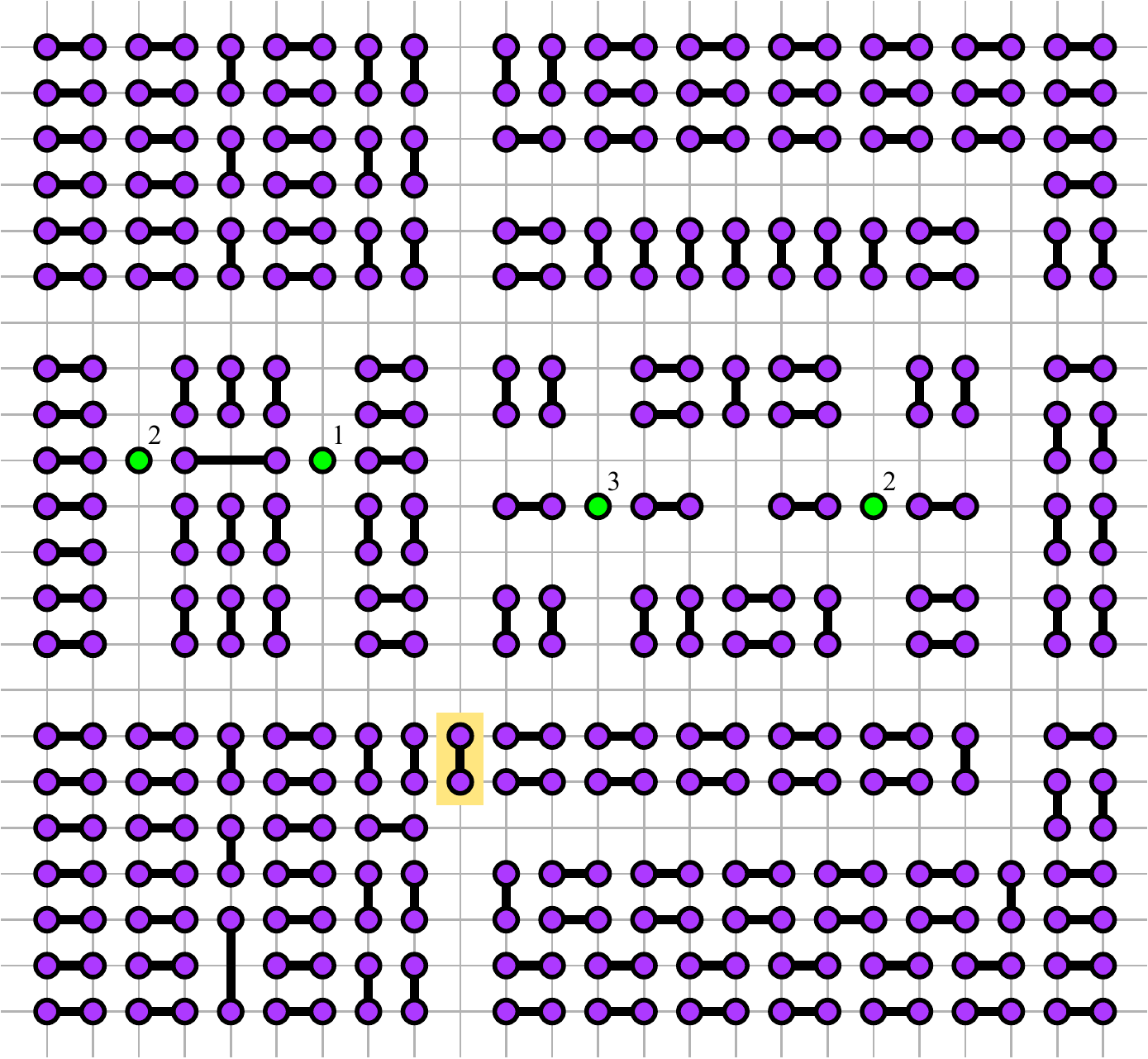}}
  \hfill
  \subfloat[Crossing (preventing) false, while still allowing true.]
  {\includegraphics[width=.48\columnwidth]{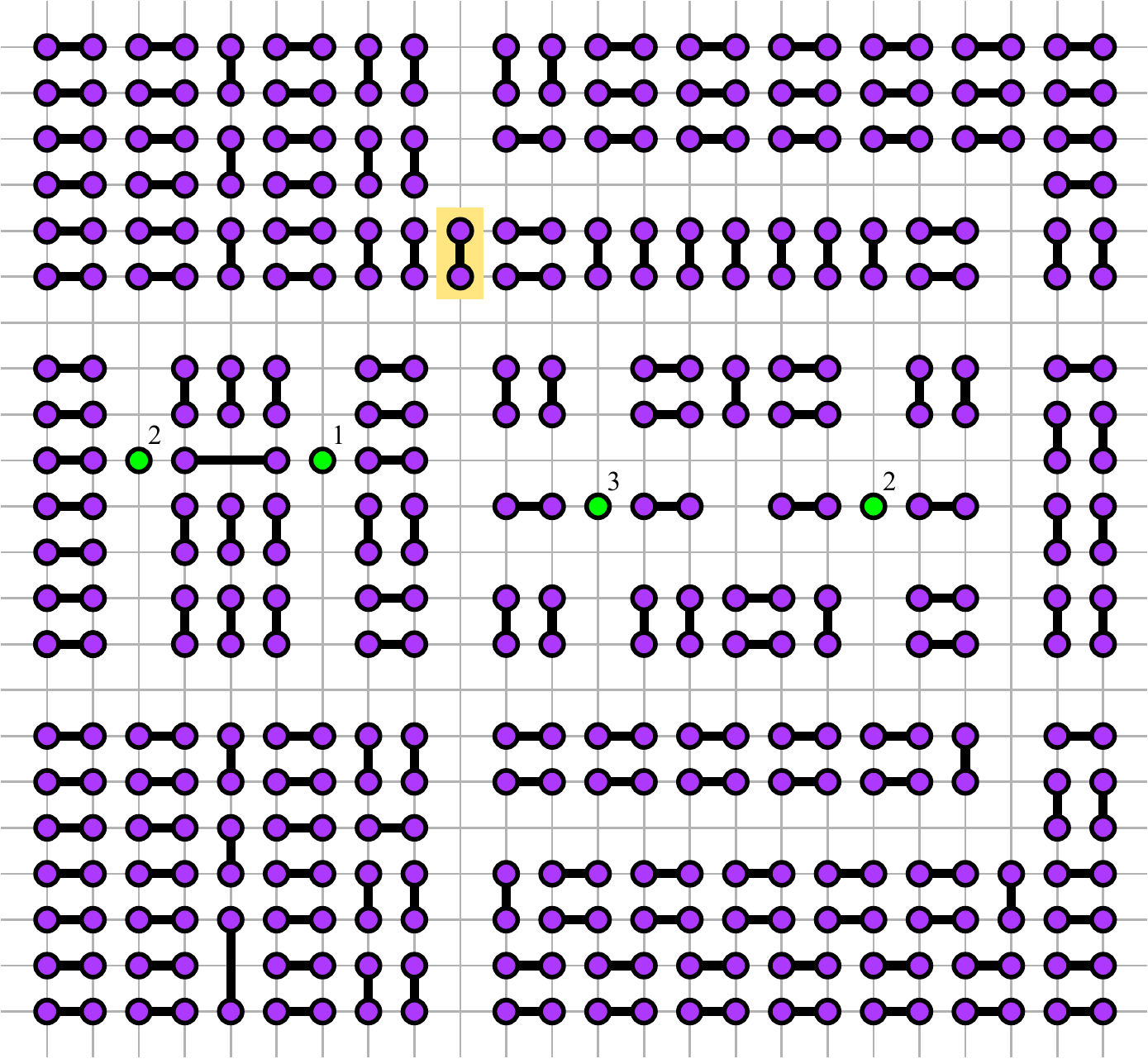}}
  \caption{\label{hardness crossing}
    Crossing gadget, with paths drawn only between obstacle pairs.
    The highlight indicates the unique added obstacle pair.}
\end{figure}

  This crossover gadget necessitates nontrivial \emph{crossing gadgets},
  because the two variable paths in the high-level picture
  (\figref{hardness overview}) have been replaced by alternation
  between two classes of paths.  \Figref{hardness crossing} illustrates
  two gadgets for crossing (preventing) the false and true settings,
  respectively.  Each of these gadgets adds just a single obstacle pair
  to the crossover gadget of \figref{hardness crossover}.
  These obstacles suffice to block the clause path in the prevented variable
  setting, but are consistent with the solutions in
  \figref{hardness crossover} both with and without the clause path,
  so they still allow the other variable setting.
  
  Finally, we can form the two-way branches on the left side $v'_i$
  and the right side $v''_i$, and the three-way branches at the
  top side $c'_j$ and the bottom side $c''_j)$, using the split gadget
  in \figref{hardness split}.  This gadget allows the incoming path
  on the top to exit at either of the two bottom ports, while still covering
  all squares, provided the exit ports both have the same parity
  (color on the checkerboard).
  A key property of the crossover and crossing gadgets (and the reason for
  the strange fifth column) is that they have even numbers of rows and columns.
  As a result, when we build a grid of these gadgets, the top entrance ports
  all have the same parity, as do all the left entrance ports.
  Therefore the split gadgets can correctly connect to the ports on the
  left, right, top, and bottom sides of the grid.
\end{proof}

\Figref{7b erik} illustrates why the NP-hardness proof
for the restricted game \cite{NumberlinkNP} does not immediately
apply to \numberwang.  The illustrated gadget,
from Fig.~7 of \cite{NumberlinkNP}, consists  of
a 1-in-4 SAT clause connected to four wires.
\Figref{7b} shows the intended solution, which has one wire in the
opposite state from the three other wires.
But \figref{erik} shows another possible solution in \numberwang,
where the wires are all in the same state (or two in each state,
depending on the definition of state).
Thus the reduction does not immediately apply to the zig-zag case.
An interesting question is whether the proof (which is rather different
from ours) can be adapted to \numberwang by modifying the gadgets
from~\cite{NumberlinkNP}.

\begin{figure}[H]
  \centering
  \subfloat[Choice 1]{\includegraphics[width=0.45\columnwidth]{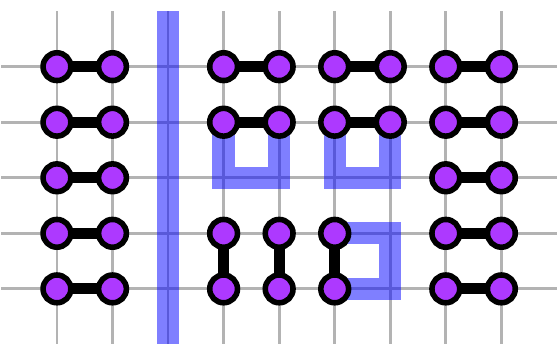}}
  \hfill
  \subfloat[Choice 2]{\includegraphics[width=0.45\columnwidth]{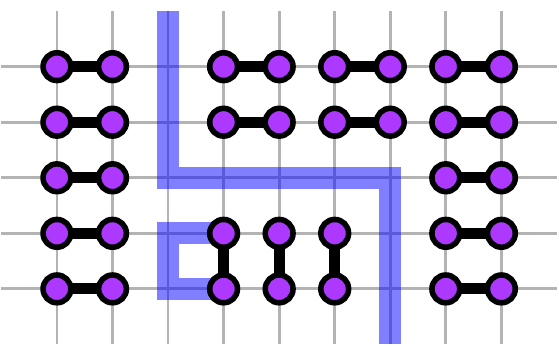}}
  \caption{\label{hardness split}
    Split gadget.}
\end{figure}

\begin{figure}[H]
  \centering
  \subfloat[\label{7b} Intended solution.]{\includegraphics[width=0.45\columnwidth]{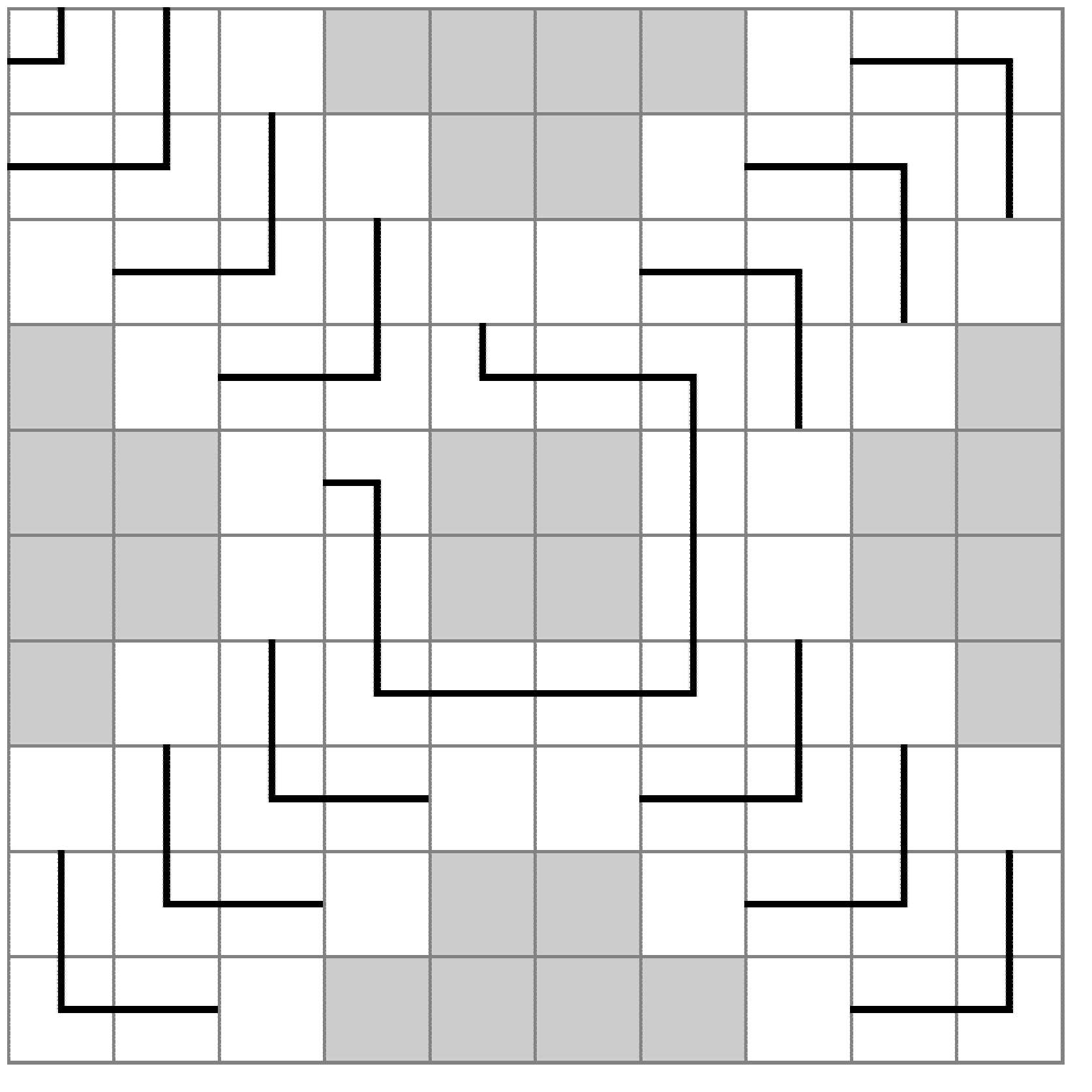}}
  \hfill
  \subfloat[\label{erik} Additional zig-zag solution.]{\includegraphics[width=0.45\columnwidth]{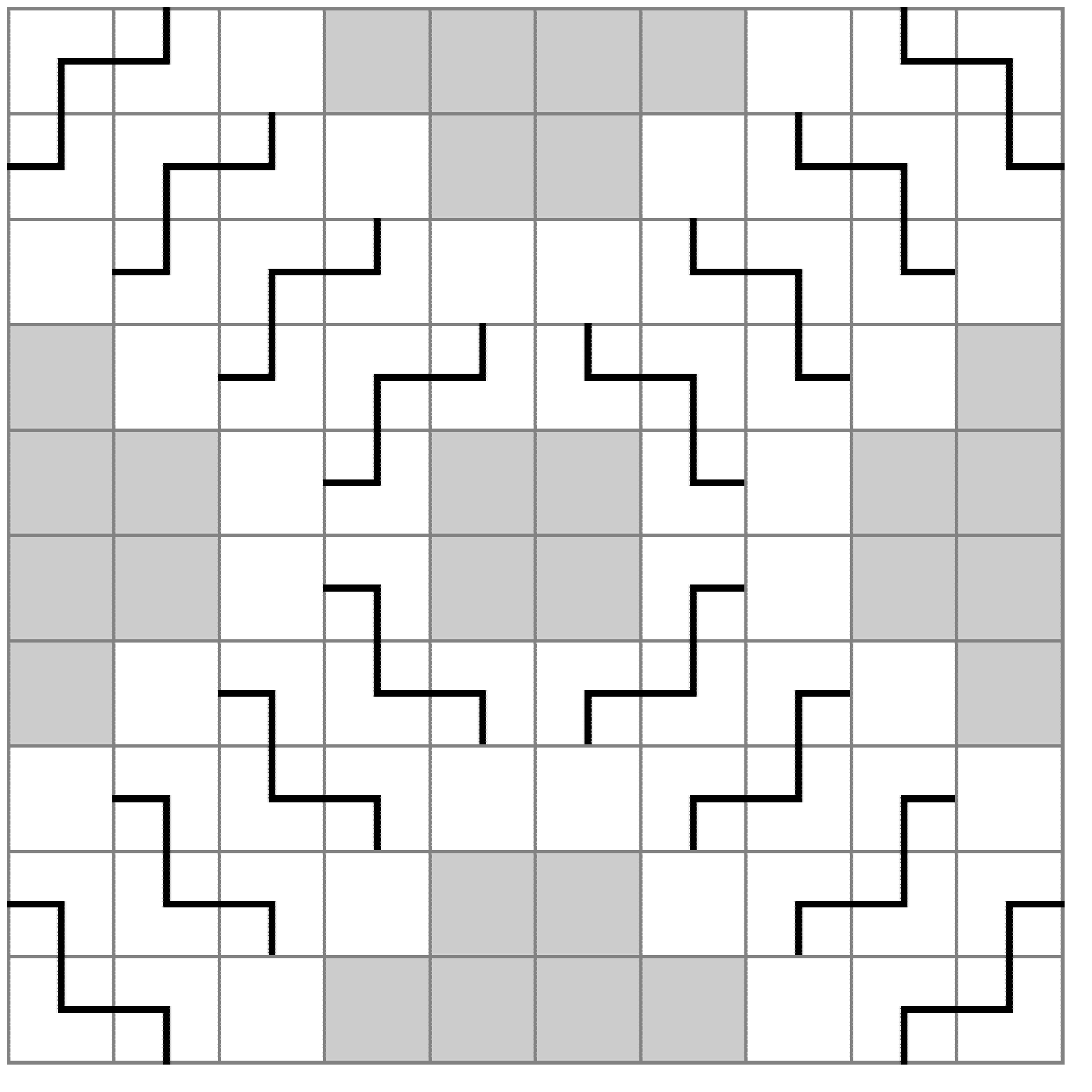}}
  \caption{\label{7b erik}
    The 1-in-4 SAT clause from \cite{NumberlinkNP}
    does not immediately work for \numberwang.
    Shaded regions represent obstacles (made by obstacle pairs).}
\end{figure}



\section{Open Questions}

There is still more fun to be had out of this game.
Is there a polynomial solution to the problem for more than one pair
of terminals?  Our reduction from \textsc{3SAT} creates instances
with a huge number of terminals. Therefore we cannot exclude the
possibility that, for any constant number of terminals, the problem can
be solved in polynomial time. If so, it would be especially interesting
to know whether the problem is fixed parameter tractable, i.e.,
solvable in $f(k) \cdot n^{O(1)}$ time for some function~$f$.
It might be that more sophisticated tools like multiterminal flow algorithms
could help answer these questions.

\begin{figure}[H] 
  \centering
  \subfloat[Crossover gadget with paths drawn only between obstacle pairs.
  \label{hardness crossover blank}]
  {\includegraphics[width=0.48\columnwidth]{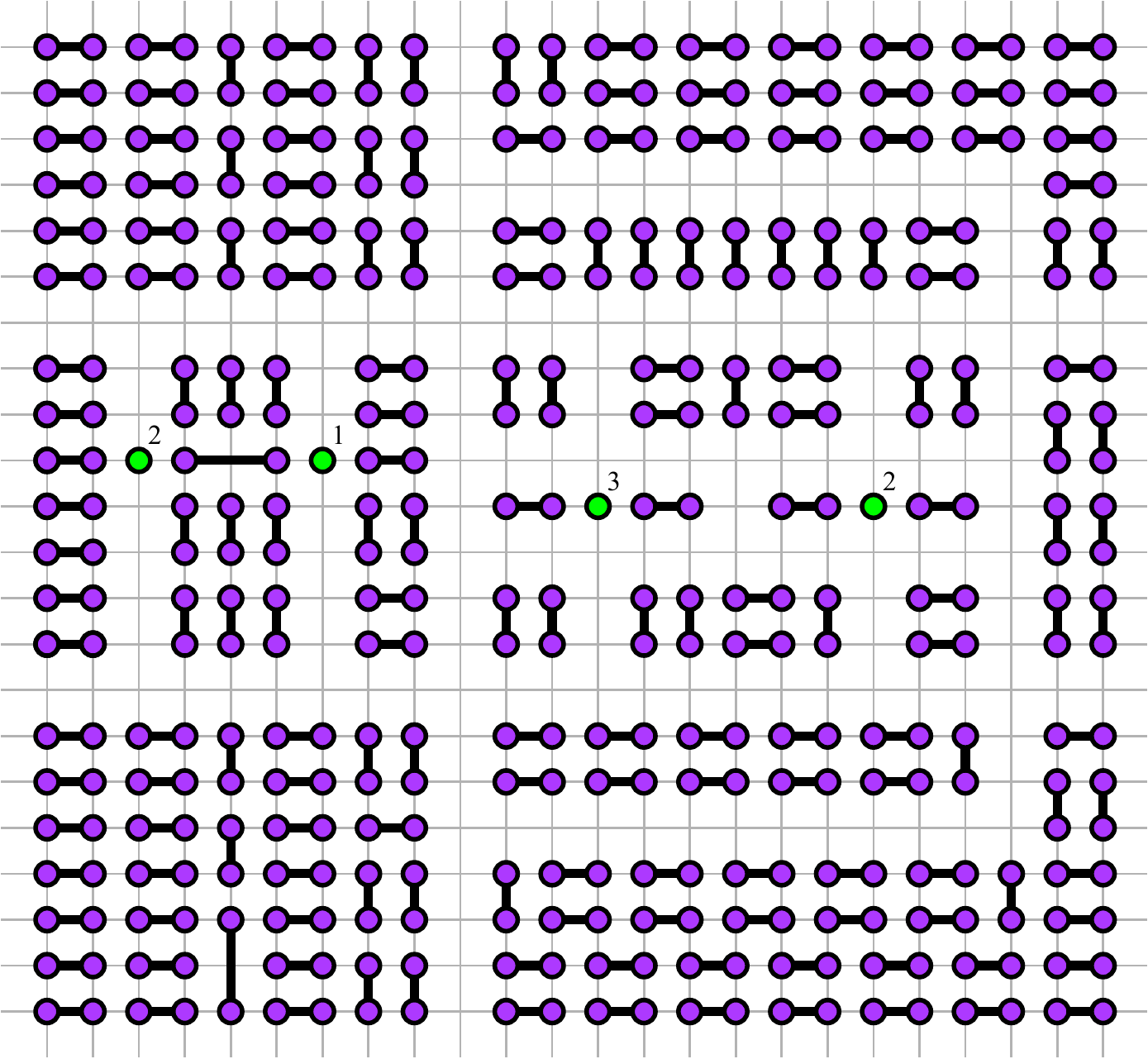}}
  
  \subfloat[False setting with clause path.]
  {\includegraphics[width=0.48\columnwidth]{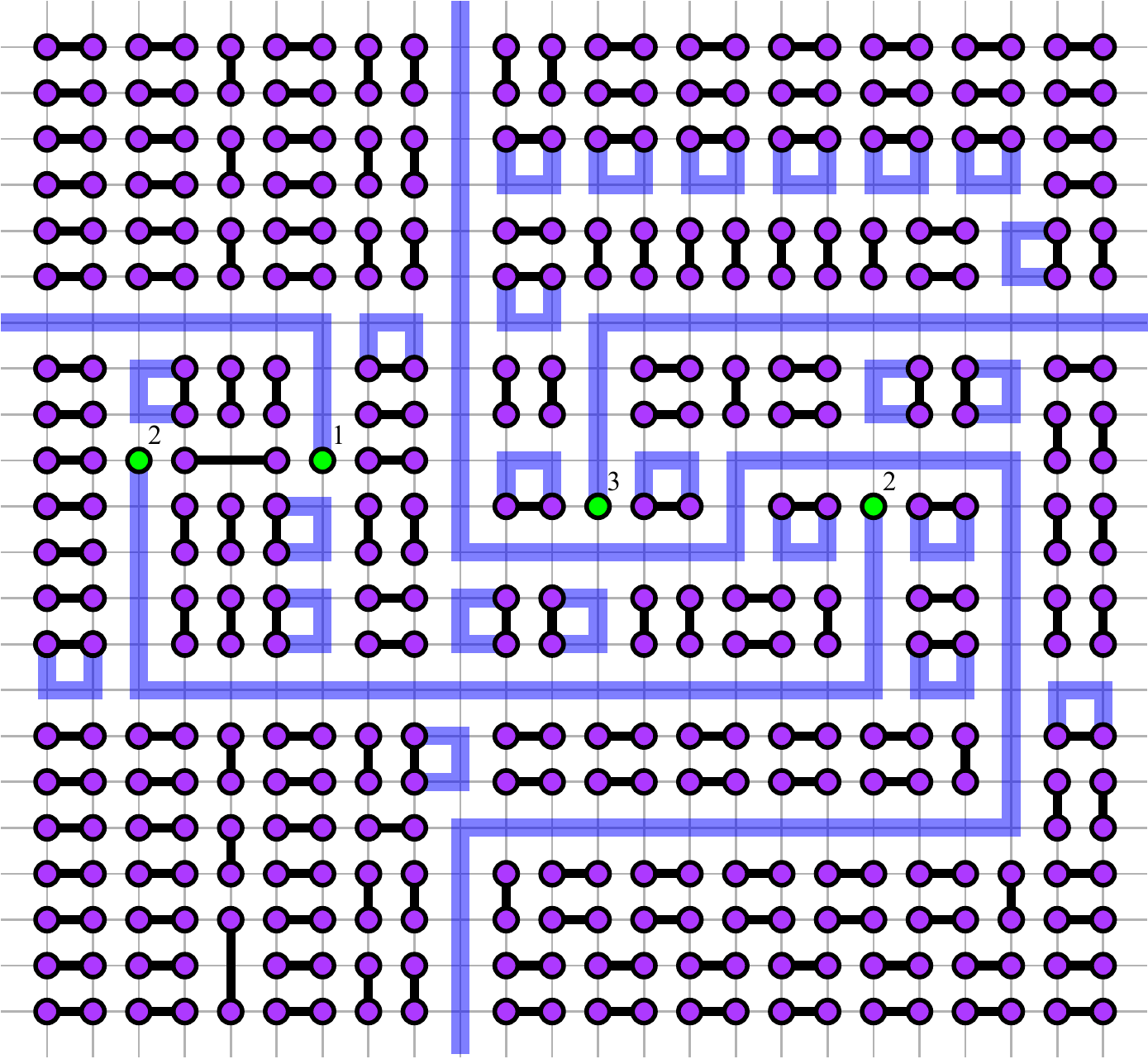}}
  \hfill
  \subfloat[False setting without clause path.]
  {\includegraphics[width=0.48\columnwidth]{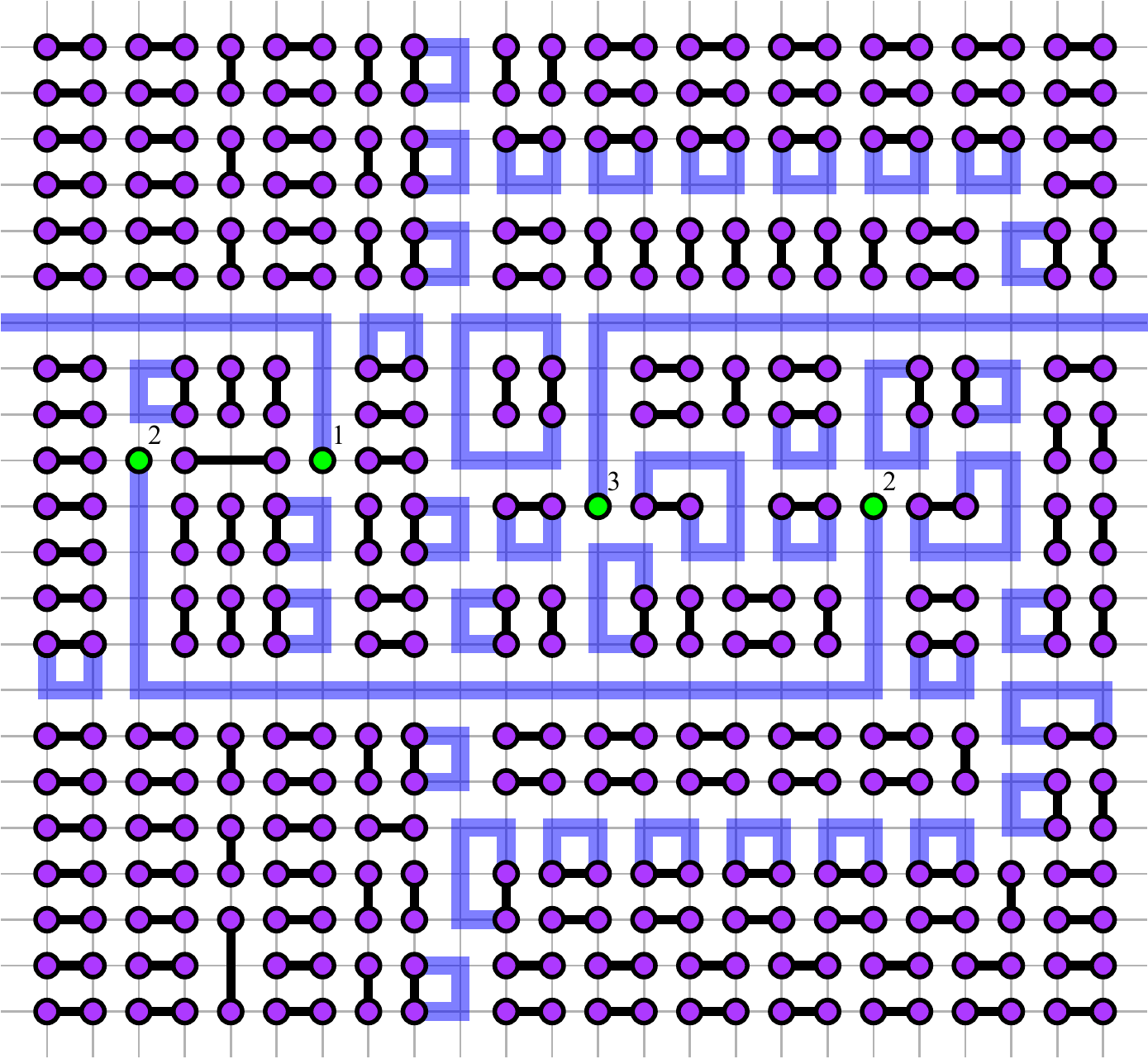}}
  
  \subfloat[True setting with clause path.]
  {\includegraphics[width=0.48\columnwidth]{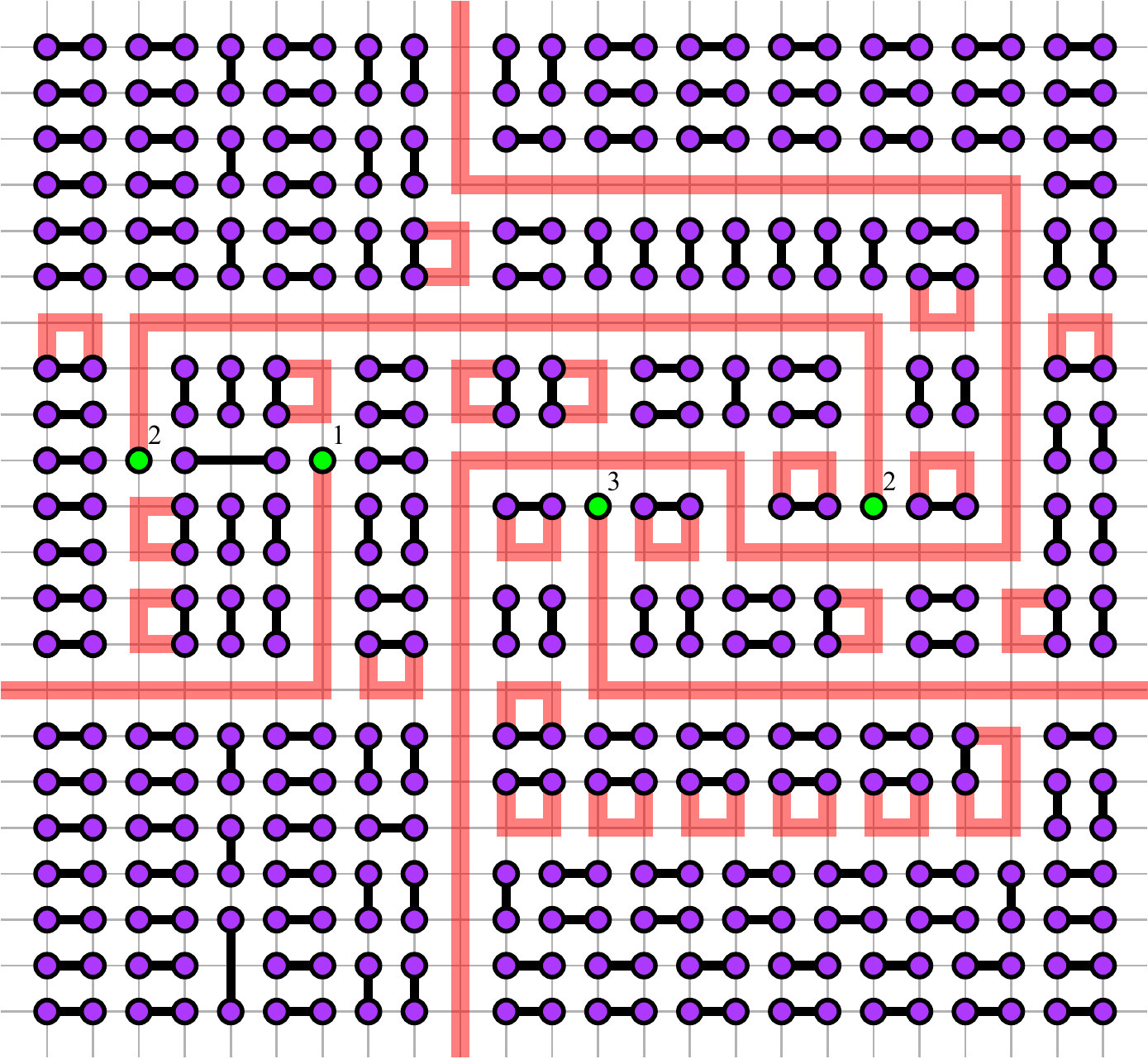}}
  \hfill
  \subfloat[True setting without clause path.]
  {\includegraphics[width=0.48\columnwidth]{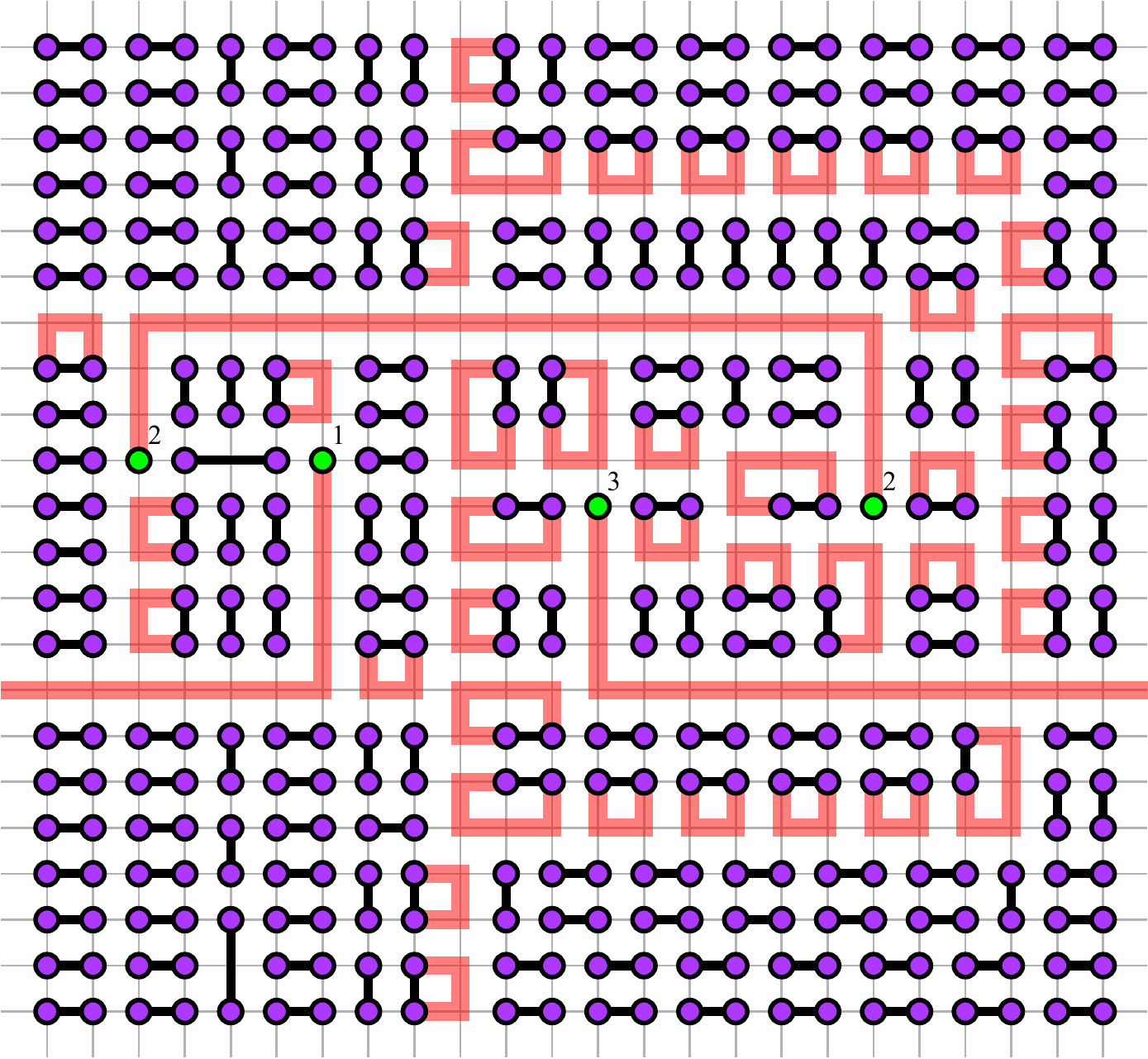}}
  \caption{\label{hardness crossover}
    Crossover gadget.}
\end{figure}

\paragraph{Acknowledgment}
This research began at the ICERM research cluster ``Towards Efficient
Algorithms Exploiting Graph Structure'', co-organized by B. Sullivan,
E. Demaine, and D. Marx in April 2014.  We thank the other participants
for providing a stimulating research environment.

The authors were introduced to Zig-Zag Numberlink through Philip Klein
(one of the participants), who in turn was introduced through his daughter
Simone Klein.  We thank the Kleins for this introduction, without which
we would not have had hours of fun playing the game and
this paper might not have existed.

E. Demaine supported in part by NSF grant CCF-1161626 and
DARPA/AFOSR grant FA9550-12-1-0423.

B. Sullivan supported in part by the National Consortium for Data Science
Faculty Fellows Program and the Defense Advanced Research Projects Agency
under SPAWAR Systems Center, Pacific Grant N66001-14-1-4063. 

Any opinions, findings, and conclusions or recommendations expressed in this publication are
those of the author(s) and do not necessarily reflect the views of DARPA, SSC
Pacific, AFOSR or the NCDS. 
 
\bibliography{./flow-game}



\end{document}